\documentclass[11pt]{article}

\usepackage[leqno]{amsmath}
\usepackage{amsthm}
\usepackage{amsfonts}
\usepackage{amssymb}
\usepackage{enumerate}
\usepackage{color}
\usepackage[pdfstartview=FitH,pdfpagemode=None,colorlinks=true,citecolor=blue,linkcolor=blue]{hyperref}

\newcommand{\eps}{\varepsilon}
\newcommand{\ket}[1]{\mbox{$\left|{#1}\right\rangle$}}

\theoremstyle{plain}
  \newtheorem{theorem}{Theorem}
  \newtheorem{lemma}{Lemma}
	\newtheorem{clm}[lemma]{Claim}

\theoremstyle{definition}
  \newtheorem{definition}{Definition}

\begin{document}
\title{Quantum lower bound for inverting a permutation with advice}
\author{
Aran Nayebi\thanks{Stanford University, \protect\url{anayebi@stanford.edu}}
\and
Scott Aaronson\thanks{MIT, \protect\url{aaronson@csail.mit.edu}}
\and
Aleksandrs Belovs\thanks{University of Latvia, \protect\url{stiboh@gmail.com}}
\and
Luca Trevisan\thanks{UC Berkeley, \protect\url{luca@berkeley.edu}}
}
\date{}

\maketitle

\begin{abstract}
Given a random permutation $f: [N] \to [N]$ as a black box and $y \in [N]$, we want to output $x = f^{-1}(y)$. Supplementary to our input, we are given classical advice in the form of a pre-computed data structure; this advice can depend on the permutation but \emph{not} on the input $y$. Classically, there is a data structure of size $\tilde{O}(S)$ and an algorithm that with the help of the data structure, given $f(x)$, can invert $f$ in time $\tilde{O}(T)$, for every choice of parameters $S$, $T$, such that $S\cdot T \ge N$. We prove a quantum lower bound of $T^2\cdot S = \tilde{\Omega}(\eps N)$ for quantum algorithms that invert a random permutation $f$ on an $\eps$ fraction of inputs, where $T$ is the number of queries to $f$ and $S$ is the amount of advice. This answers an open question of De et al. \newline
\indent We also give a $\Omega(\sqrt{N/m})$ quantum lower bound for the simpler but related Yao's box problem, which is the problem of recovering a bit $x_j$, given the ability to query an $N$-bit string $x$ at any index except the $j$-th, and also given $m$ bits of classical advice that depend on $x$ but not on $j$.
\end{abstract}

\section{Introduction}
In defining notions of security for finite functions, a standard approach to analyzing the complexity of inverting a fixed function (instead of a family of functions) is to use running-time + program-length as a measure of complexity. If one wants to invert a random permutation or function uniformly (that is, given no advice) on all inputs, then the \emph{classical} lower bound $T \ge N$ (ignoring lower-order factors) holds. A quantum computer, however, can achieve $T = O(\sqrt{N})$ (by Grover's search algorithm \cite{grover}), which is optimal by Ambainis' result \cite[Theorem 6.2]{ambainis}. Furthermore, for inverting a random permutation with advice, Hellman \cite{hellman} showed that classically, every one-way permutation can be inverted in time $\tilde{O}(\sqrt{N})$ using a data structure of size $\tilde{O}(\sqrt{N})$. The question that naturally arises is if there is a similar trade-off when one has a quantum algorithm and is looking at running-time + program-length as a measure of complexity. Hence, our results are motivated by the following question raised by De et al. \cite[pg. 3]{de}: ``What is the complexity of inverting a random permutation [...] with a quantum computation that takes advice?'' We show at least that whatever gain, if any, can be obtained by using long programs is a polynomial gain and not an exponential gain.\newline
\indent In De et al. \cite{de} (and going back to ideas of Yao \cite{yao}) the classical lower bound of $S\cdot T = \tilde{\Omega}(\eps N)$ is proved by showing that permutations that can be inverted given some advice with few queries have a small representation given the advice, and hence occur with low probability when picked randomly. Here we are able to find such a small representation even for permutations that can be inverted with few \emph{quantum} queries.\newline
\indent The common quantum lower bound techniques are the polynomial method (introduced in \cite{poly}) and the adversary method (introduced in \cite{ambainis}) - both of which have been successfully applied to many problems. We will refrain from discussing these techniques in detail, but we only mention them to point out that the lower bounds we obtain rely on a precursor to the adversary method known as the hybrid argument (which we describe in \S 2). We should note that using techniques other than the commonly used lower bound techniques for advised quantum computations is not new. Nishimura and Yamakami \cite{NY} use what they call the ``algorithmic method'' to obtain a quantum lower bound for solving the ``multiple-block'' ordered search problem (where the algorithm takes advice). However, their assumption is that the quantum algorithm makes \emph{nonadaptive} queries, and in fact, it is not clear how to extend their lower bound argument to the case that the algorithm makes adaptive queries \cite{pc} (which is the case when an algorithm inverts a random permutation with advice).\newline
\indent Specifically, the results we prove are the following:
\begin{theorem}[Yao's box problem lower bound]\label{result1}
Suppose we have $N$ boxes, each of which contains a Boolean bit. We will label the boxes from 1 to $N$. Given as input an integer $j$ such that $1 \le j \le N$, $m$ bits of classical advice $($where the advice can depend on the box contents but not on $j)$, and the bit pattern of the $N$ boxes as an oracle, any quantum algorithm $\mathcal{A}$ will need $T = \Omega(\sqrt{N/m})$ queries $($where $\mathcal{A}$ is not allowed to query the bit of the $j$-th box$)$ in order to output the bit of the $j$-th box $($with error probability at most $1/3)$.
\end{theorem}
\begin{theorem}[Permutation inversion lower bound, informal]\label{result2}
Given a random permutation $f\colon [N]\to [N]$ as an oracle, classical advice $\alpha = \alpha(f)$ of at most $S$ bits, and some $y \in [N]$ as input, we want to output $f^{-1}(y)$. Any quantum algorithm $\mathcal{A}$  satisfies a $T^2\cdot S = \tilde{\Omega}(\eps N)$ trade-off lower bound\footnote{The notation $\tilde{\Omega}$ hides factors polynomial in $\log N$ and $\log \frac1\eps$.} where $\mathcal{A}$ makes at most $T$ oracle queries to $f$, and solves the problem on an $\eps$ fraction of inputs. 
\end{theorem}
The best known upper bounds are the classical ones. For the box problem, this is $O(N/m)$ as proven by Yao \cite{yao}. For inverting a permutation with advice, Hellman \cite{hellman} showed that there is a data structure of size $\tilde{O}(S)$ and an algorithm that with the help of the data structure, given $f(x)$, can invert $f$ in time $\tilde{O}(T)$, for every choice of parameters $S$, $T$, such that $S\cdot T \ge N$. Finally the lower bounds in the classical case are ${\Omega}(N/m)$ for the box problem, and $S\cdot T = \tilde{\Omega}(\eps N)$ for the problem of permutation inversion \cite{yao}.

\section{Preliminaries}
We will be working in the standard black box model of computation (cf. \cite[\S 8]{qcnotes}) where the goal is to compute a Boolean function $f: \{0,1\}^n \to \{0,1\}$ given an input $x \in \{0,1\}^n$. However, $x$ is not given explicitly, but is instead given as a black box, and we are being charged unit cost for each \emph{query} that we make. Note that $x$ here takes the same place as $f$ in the introduction (\S 1). So an equivalent view is that $f$ is given as a black box and we are being charged unit cost for each query that we make to $f$. The reason why we want to query a black box is that in order to obtain a lower bound on the number of queries $T$ the quantum algorithm would need to make, we are not exploiting any specific weakness of a particular implementation.\newline
\indent Since the hybrid argument will be the machinery for the lower bounds proven here, we will discuss this technique in more detail. It is based on the intuition that an adversary runs the algorithm with one input and, after that, changes the input slightly so that the correct answer changes but the algorithm does not notice that. It was first introduced by Bennett et al. \cite{bbbv} in 1997 to prove a tight quantum lower bound on search of $T = \Omega(\sqrt{N})$ and a (suboptimal) quantum lower bound of $T = \Omega(\sqrt[3]{N})$ for inverting a random permutation on all inputs (with no advice). A clear exposition is given by Vazirani \cite{vazirani}, which we will base our discussion on.\newline
\indent More formally, fix a quantum algorithm $\mathcal{A}$. The state of $\mathcal{A}$ between successive queries can be written as
\begin{equation*}
\ket{\phi} = \sum_{c}\alpha_c\ket{c},
\end{equation*}
where $c$ runs over the computational basis states. During each query, each basis state $c$ queries a particular bit position of the input. The main question is: how sensitive is the output of $\mathcal{A}$ to the modification of the input $x$ on a few bit positions (or, in the equivalent view, how sensitive is the output of $\mathcal{A}$ to the modification of the output of $f$ on a few inputs)? \newline
\indent To answer this question, we introduce some definitions. In what follows, we will maintain the first view that we are querying an input $x \in \{0,1\}^n$, though all these definitions can be easily modified in the equivalent view that we are instead querying $f$. Which view is more satisfactory is simply a matter of convenience; for instance, in the proof of Theorem~\ref{result1} in \S 3.3 we maintain the first view, and in the proof of Theorem~\ref{result2} in \S 4.2 we maintain the second view.
\begin{definition}\label{d1}
\emph{The \emph{query magnitude} at bit position $j$ of $\ket{\phi} = \sum_{c}\alpha_c\ket{c}$ is defined to be $q_j(\ket{\phi}) = \sum_{c \in C_j}|\alpha_c|^2$, where $C_j$ is the set of all computational basis states that query bit position $j$.}
\end{definition}
Now, suppose $\mathcal{A}$ runs for $T$ steps on an input $x \in \{0,1\}^n$. Then the run of $\mathcal{A}$ on input $x$ can be described as a sequence of states $\ket{\phi_0}, \ldots, \ket{\phi_T}$, where $\ket{\phi_t}$ is the state of $\mathcal{A}$ before the ${t+1}^{\textrm{st}}$ query.
\begin{definition}\label{d2}
\emph{The \emph{total query magnitude} at bit position $j$ of $\mathcal{A}$ on input $x$ is defined to be:
\begin{equation}\label{qm}
q_j(x) = \sum_{t = 0}^{T-1}q_j(\ket{\phi_t}).
\end{equation}}
\end{definition}
Let $[n] = \{1,\ldots, n\}$. By definition, since $\mathcal{A}$ makes $T$ queries and each $\ket{\phi_t}$ is of unit length, then
\begin{equation}\label{ubqm}
\sum_{j\in [n]}q_j(x) \le T.
\end{equation}
\indent Now, if the total query magnitude of bit position $j$ is very small, then $\mathcal{A}$ cannot distinguish whether the input $x$ is modified by flipping its $j$-th bit:
\begin{lemma}[``Swapping lemma'']\label{swap-lemma}
Let $\ket{\phi_x}$ and $\ket{\phi_y}$ denote the final states of $\mathcal{A}$ on inputs $x$ and $y$, respectively. Then:
\begin{equation}\label{swap}
\|\ket{\phi_x} - \ket{\phi_y}\| \le \sqrt{T\sum_{j \in \Delta(x,y)}q_j(x)},
\end{equation}
where $\Delta(x,y) = \{j:\mbox{    }x_j\ne y_j\}$, and $\|\ket{\phi_x} - \ket{\phi_y}\|$ denotes the Euclidean distance between the two vectors.
\end{lemma}
We will not give a proof of Lemma~\ref{swap-lemma} here (it is stated and proven as Lemma 3.1 of Vazirani \cite[pp. 1764-65]{vazirani}), but will instead apply it in the following sections to prove our main results.\newline
\indent We will also make use of the following well-known fact (first proven by Bernstein and Vazirani \cite[Lemma 3.6, pg. 1421]{be}):
\begin{lemma}\label{plemma}
Let $\mathcal{D}(\psi)$ denote the probability distribution that results from a measurement of $\ket{\psi}$ in the computational basis. If $\|\ket{\phi}-\ket{\psi}\| \le \eps$, then ${\|\mathcal{D}(\phi) - \mathcal{D}(\psi)\|}_{1}\le 4\eps$.
\end{lemma}
The notation ${\|\mathcal{D}(\phi) - \mathcal{D}(\psi)\|}_{1}$ denotes the total variation distance\footnote{For two probability distributions $P$ and $Q$, their total variation distance is $\sum_{x}|P(x) - Q(x)|$. Hence, since $\ket{\phi} = \sum_c\alpha_c\ket{c}$ and $\ket{\psi} = \sum_c\beta_c\ket{c}$, then ${\|\mathcal{D}(\phi) - \mathcal{D}(\psi)\|}_{1} = \sum_c|\alpha^2_c - \beta^2_c|$.} between the probability distributions $\mathcal{D}(\phi)$ and $\mathcal{D}(\psi)$. Thus, Lemma~\ref{plemma} states that if two unit-length superpositions are within Euclidean distance $\eps$ then observing the two superpositions gives samples from distributions whose 1-norm or total variation distance is at most $4\eps$.

\section{Yao's box problem}

\subsection{Classical lower bound}
In his study of a classical lower bound for inverting a random permutation with advice, Yao \cite{yao} introduced a simpler (but related) problem to analyze, which we will refer to as the ``box problem'', stated as follows:\newline
\begin{quote}
\small{``Let $N, m$ be positive integers. Consider the following game to by played by $A$ and $B$. There are $N$ boxes with lids $\textrm{BOX}_1, \textrm{BOX}_2,\ldots, \textrm{BOX}_N$ each containing a Boolean bit. In the preprocessing stage, Player $A$ will inspect the bits and take notes using an $m$-bit pad. Afterwards, Player $B$ will ask Player $A$ a question of the form ``What is in $\textrm{BOX}_i$?''. Before answering the question, Player $A$ is allowed to consult the $m$-bit pad and take off the lids of an adaptively chosen sequence of boxes not including $\textrm{BOX}_i$. The puzzle is, what is the minimum number of boxes $A$ needs to examine in order to find the answer?'' }\cite[pg. 84]{yao}.\newline
\end{quote}
On a classical computer, Yao proves that $T \ge \lceil N/m \rceil - 1$, which is optimal. In fact, one strategy that $A$ may adopt is divide-and-conquer. Namely, divide the boxes into $m$ consecutive groups each containing no more than $\lceil N/m \rceil$ members. In the preprocessing stage, $A$ records for each group $g$ the parity $a_g$ of the bits in that group. Then, to answer the query ``What is in $\textrm{BOX}_i$?'', $A$ can lift the lids of all the boxes (except $\textrm{BOX}_i$) in the group $k$ containing $\textrm{BOX}_i$ and compute the parity $b$ of these lids. Clearly, the bit in $\textrm{BOX}_i$ is $a_k \oplus b$. Therefore, $A$ never needs to lift more than $\lceil N/m \rceil - 1$ lids with this strategy.

\subsection{Yao's advice model}
We define the model more precisely, in a similar manner as Yao \cite[pg. 85]{yao} does. Let $\{0,1\}^N$ be the set of all possible $2^N$ bit patterns of the $N$ boxes. Fix our $m$-bit classical advice string $\alpha \in \{0,1\}^m$, where $1 \le m < N$. This induces a partition $D_{\alpha}\subseteq \{0,1\}^N$ and $N$ \emph{partial} Boolean functions $f_{j,\alpha}$, where $f_{j,\alpha}$ is the partial Boolean function that computes the bit of the $j$-th box from the remaining $N-1$ boxes given advice string $\alpha$. Note that
\begin{equation*}
2^N \le \sum_{\alpha \in \{0,1\}^m}|D_{\alpha}| \le 2^m\cdot 2^N,
\end{equation*}
since for any two $m$-bit advice strings, $\alpha'$ and $\alpha''$, it is \emph{not} necessarily the case that $D_{\alpha'} \cap D_{\alpha''} = \emptyset$. Thus, by a standard averaging argument, there is a partition $D_\alpha$ such that $2^{N-m} \le |D_{\alpha}| \le 2^N$. For advice strings $\alpha$ such that $|D_{\alpha}| < 2^{N-m}$, we can forgo consideration of such advice strings as they are not ``useful'', since $2^{N-m}$ is already exponentially small in $m$.  \newline
\indent Although $D_{\alpha}$ is a set consisting of arbitrary strings and so it does not have much structure to it that we can exploit (in order to get a lower bound using the polynomial or adversary methods), it turns out that this lower bound on the size of $D_{\alpha}$ is enough to prove a quantum lower bound using the hybrid argument, as we will show in \S 3.3.

\subsection{Proof of Theorem~\ref{result1}}
First, we prove the following combinatorial lemma:
\begin{lemma}\label{L1}
Suppose $D \subseteq \{0,1\}^n$ of size $|D| \ge 2^{n-m}$ and that we randomly select a set $I$ of $m + 1$ indices. Then with probability 1 there are at least two strings $x$ and $y$ in $D$ that differ only in a subset of the coordinates in $I$.
\end{lemma}
\begin{proof}
There are only $2^{n-(m + 1)}$ different ways of fixing the coordinates not in $I$, and since this is less than the number of elements of $D$, there must be two or more elements of $D$ that are identical outside of $I$.
\end{proof}
The rest of the proof of Theorem~\ref{result1} follows easily, by using the same techniques as in the proof of Theorem 3.3 of Vazirani \cite{vazirani}. Fix a $j$ such that $1 \le j \le N$, and let $\mathcal{A}$ be the quantum algorithm that uses the advice string $\alpha$ to compute $f_{j,\alpha}$ given the bit pattern of the $N$ boxes as an oracle. By Lemma~\ref{L1}, let $I$ be a randomly selected set of $m + 1$ indices, and let $x$ and $y$ be two strings in $D_{\alpha}$ that differ only in a subset of the coordinates in $I$ (note that since $\mathcal{A}$ is not allowed to query the $j$-th bit position, then we will set the total query magnitude at the $j$-th bit position to be 0). Now, we make the following general observation about total query magnitude: Let $z$ be chosen uniformly at random among the $N-1$ bit positions other than the $j$-th bit position of the oracle $x \in \{0,1\}^{N}$ (which represents the bit pattern of the $N$ boxes). Then for all $t \in \{0,\ldots, T-1\}$, if we let $\ket{\phi_t}$ be the state of the algorithm $\mathcal{A}$ before the $t+1^{\textrm{st}}$ query to $x$, then $\mathbb{E}[q_z(\ket{\phi_t})] = \frac{1}{N-1}$. By linearity of expectation, we have that $\mathbb{E}[q_z(x)] = \frac{T}{N-1}$.\newline\newline
\indent Therefore, by Lemma~\ref{swap-lemma} and a standard averaging argument,
\begin{equation} \label{2}
\begin{split}
\|\ket{\phi_x}-\ket{\phi_y}\| &\le \sqrt{T\sum_{z \in I}q_z(x)}\\
&\le \sqrt{T\left(\frac{T(m + 1)}{N-1}\right)} = T\sqrt{\frac{m + 1}{N-1}}.
\end{split}
\end{equation}
Assume that $\mathcal{A}$ errs with probability bounded by $1/3$. Thus, ${\|\mathcal{D}(\phi_x)-\mathcal{D}(\phi_y)\|}_{1} \ge 1/3$ by footnote 2. If $T < (1/12)\sqrt{\frac{N-1}{m + 1}}$, by Lemma~\ref{plemma}, it would follow that ${\|\mathcal{D}(\phi_x)-\mathcal{D}(\phi_y)\|}_{1} < 1/3$, contradicting the bound on the error probability of $\mathcal{A}$. Therefore, $T \ge (1/12)\sqrt{\frac{N-1}{m + 1}}$, as desired.
\subsection{Quantum upper bound?}
The next question that arises is whether we can also obtain a quantum algorithm that solves this problem in $O(\sqrt{N/m})$ time? At the moment, this is an open question. If we use the same advice as Yao (namely, the parity of bits in $m$ groups), then it seems likely that we are reduced to solving parity. However, computing the parity of $N$ bits takes $\Omega(N)$ time on a quantum computer, which Beals et al. \cite{poly} prove using the polynomial method.\newline
\indent A reasonable approach is to think that an improvement over the classical running time of $O(N/m)$ for the box problem would involve changing the advice. However, even this does not seem viable. Say we change our advice to the number of boxes with a 1 in them. Now, suppose our advice says that there are a total of $r$ boxes with ones in them.  If we are asked to find the value in box $i$, we know that the number of ones in the remaining $N-1$ boxes is either $r$ or $r-1$.  To find out the value of box $i$, we just have to figure out which case we are in. Let $t$ be the number of boxes with a 1 in them (from the $N-1$ remaining boxes). Then in $O(\sqrt{N/t})$ time on a quantum computer we can estimate $t$ with high probability within a small error, which is optimal (cf. Brassard et al. \cite[Corollary 3]{brassard}). But in our problem, the only issue is that we want to know \emph{exactly} what $t$ is (since a difference between $r-1$ or $r$ is crucial). Of course, in Corollary 4 of Brassard et al. \cite{brassard}, we can get an exact estimate of $t$ with high probability in $O(\sqrt{N\cdot t})$ time, which is also optimal. But if $t \ge N/2$ (which is a good estimate for the number of ones in a random string), we do not get an asymptotic improvement over the classical run-time of $O(N)$.

\section{Inverting a random permutation with advice}

\subsection{Classical lower bound}
The space-time complexity of inverting a random permutation with advice is well-studied in the classical setting. We are given oracle access to a bijection $f: [N]\to [N]$, $y = f(x)$ as input, and a data structure of size $S$ (which will be the advice), and we will output $x$. For our purposes, it will be convenient to set $N = 2^n$ and identify $\{0,1\}^n$ with $[N]$, as De et al. \cite[pg. 1]{de} do. In 1980, Hellman \cite{hellman} proved that for every choice of parameters $S, T$ such that $S \cdot T \ge N$, there is a data structure of size $\tilde{O}(S)$ and an algorithm that with the help of the data structure, given $f(x)$, which is a permutation, can invert $f$ in time $\tilde{O}(T)$. Specifically, when $S = T$, every one-way permutation can be inverted in time $\tilde{O}(\sqrt{N})$ using a data structure of size $\tilde{O}(\sqrt{N})$ as advice. Yao \cite[pg. 86]{yao} proves Hellman's tradeoff to be optimal. In fact, one can generalize Yao's arguments to show that there are permutations for which the amount of advice $S$ and the oracle query complexity $T$ must satisfy
\begin{equation*}
S\cdot T = \tilde{\Omega}(\eps N),
\end{equation*}
for any classical algorithm that inverts on an $\eps$ fraction of inputs \cite[\S 2.3]{de}. 
\newline
\indent To explain Hellman's construction (for the case $S = T$), we suppose for simplicity that $f$ is a cyclic permutation and $N = s^2$ is a perfect square. Then our advice will be a data structure of pairs $(x_i, x_{i+1})$ of $\sqrt{N}$ ``equally spaced'' points $x_1,\ldots,x_s$, such that $x_{i +1} = f^{(s)}(x_i)$, where the notation $f^{(s)}(x_i)$, means $f$ has been iterated $s$ times. Given an input $y$ to invert, we compute $f(y), f(f(y))$, and so on, until for some $j$ we reach a point $f^{(j)}(y)$ in the advice. Then we read the value $f^{(j-s)}(y)$, and by repeatedly computing $f$ we eventually reach $f^{(-1)}(y)$. This takes $O(s)$ evaluations of $f$ and lookups in the advice, so both time and advice complexity are approximately $O(\sqrt{N})$. Note that if the permutation $f$ is not cyclic, we can have a similar construction for each cycle of length less than $s$, and if $N$ is not a perfect square, we can round $s$ up to $\lceil \sqrt{N} \rceil$.

\subsection{Proof of Theorem~\ref{result2}}
\renewcommand{\Pr}{\mathop{\mathbb{P}}}
First, let us specify the problem in more detail.
A quantum algorithm $\mathcal{A}$ is given oracle access to a permutation $f\colon [N]\to[N]$, classical advice $\alpha = \alpha(f)$ of at most $S\ge 1$ bits, and an input element $y\in [N]$.  The algorithm makes at most $T$ queries to $f$.
We say that $\cal A$ {\em inverts} $y\in [N]$ in $f$ iff it outputs $f^{-1}(y)$ with probability at least $2/3$.
We are given that
\begin{equation}
\label{dano}
\Pr_{f,y} [\text{$\cal A$ inverts $y$ in $f$}] \ge \eps,
\end{equation}
where a permutation $f\colon[N]\to[N]$ and $y\in[N]$ are chosen uniformly at random.  In this section, we prove that, in this case, the inequality 
\begin{equation}
\label{nado}
T^2\cdot S = \tilde\Omega(\eps N)
\end{equation}
holds.  
The proof is similar to the proof in Section 10 of~\cite{de}.  
We use the quantum algorithm $\cal A$ to compress permutations, and then apply the following lemma:

\begin{lemma}[{\cite[Fact 10.1]{de}}]\label{l1}
Let $X$ and $Y$ be finite sets, and $R$ be a finite probability distribution.
Suppose there exist a randomized encoding procedure $E\colon X \times R \to Y$, and a decoding procedure $D\colon Y\times R \to X$ such that, for all $x\in X$:
\begin{equation*}
\underset{r \sim R}{\mathbb{P}}[D(E(x,r),r) = x] \ge c,
\end{equation*}
then $|Y|\ge c|X|$. 
\end{lemma}
\begin{proof}
The proof is short, and we give it for completeness.
By a standard averaging argument, there is an $r$ such that for at least a $c$ fraction of the $x$'s, we have $D(E(x,r),r) = x$. 
This means that $E(x,r)$ must attain at least $c |X|$ values, giving the desired inequality.
\end{proof}

Proceeding towards the proof, we may assume that
\begin{equation}
\label{T2}
T^2 \le C\eps N
\end{equation}
for some constant $C$ to be defined later, since otherwise the statement~\eqref{nado} is trivial.
Next, let $\eps'=\eps/2$.
Eq.~\eqref{dano} implies that there exists a set $X$ of permutations such that $X$ has size at least $\eps' N$, and
\[
\Pr_{y} [\text{$\cal A$ inverts $y$ in $f$}] \ge \eps'
\]
for all $f\in X$.

\begin{lemma}\label{l2}
Let $\mathcal{A}$ and $X$ be as defined above. 
Then there exist a randomized encoding procedure $E\colon X \times R \to Y$, and a decoding procedure $D\colon Y\times R \to X$ such that, for all $f\in X$:
\begin{equation*}
\underset{r \sim R}{\mathbb{P}}[D(E(f,r),r) = f] \ge 0.8,
\end{equation*}
and
\begin{equation*}
\log |Y| \le \log N! - \Omega\left(\frac{\eps N}{T^2}\right) + S + O(\log N).
\end{equation*}
\end{lemma}

\begin{proof}
First we describe the construction of $E(f,r)$.
Fix a permutation $f\in X$.
Let $I$ be the set of elements $x\in[N]$ such that $\cal A$ inverts $f(x)$.  By our definition of $X$,
\begin{equation}
\label{condition}
 |I| \ge \eps' N.
\end{equation}
Let $q_z(x)$ denote the total query magnitude of $z \in [N]$ when the algorithm $\mathcal{A}$ is executed with advice $\alpha = \alpha(f)$, oracle access to $f$, and input $y=f(x)$.

Randomness of the encoding subroutine is given by a random subset $R \subseteq [N]$ with each element of $[N]$ independently chosen to be in $R$ with probability $\delta/T^2$, where $0 < \delta < 1$ is some constant to be specified later.

For each $x\in I$, consider the following two random events:
\begin{equation}
\label{AandB}
(A)\quad x \in R \qquad\quad\mbox{and}\qquad\quad (B)\quad \sum_{z \in R\setminus\{x\}}q_{z}(x) \le \frac cT,
\end{equation}
where $c$ is some constant to be specified later.
Note that events (A) and (B) are independent, since event (A) depends on whether $x \in R$ and event (B) depends on whether $z \in R$, where $z \ne x$.
We say that an element $x\in I$ is {\em good} if it satisfies both (A) and (B).
Let $G$ denote the set of good elements.

\begin{clm}
With probability at least 0.8 over the choice of $R$, we have $|G| = \Omega(\eps N/T^2)$. 
\end{clm}

\begin{proof}
Let $H = R\cap I$.  This is just a binomial distribution.
The expected value of $|H|$ is ${|I|\delta}/{T^2}$.  
By the multiplicative Chernoff bound and~\eqref{T2}:
\begin{equation} \label{H}
\underset{R}{\mathbb{P}}\left[|H| \ge \frac{|I|\delta}{2T^2}\right] 
\ge 1-\exp\Bigl(-\frac{|I|\delta}{8T^2}\Bigr) 
\ge 1-\exp\left( - \frac{\delta}{16C} \right)
\ge 0.9,
\end{equation}
if $C$ is small enough.
Next, by linearity of expectation and by \eqref{ubqm},
\[
\underset{R}{\mathbb{E}}\left[\sum_{z \in R\setminus\{x\}}q_{z}(x)\right]
= \sum_{z \in [N]\setminus\{x\}}\frac{\delta}{T^2}\; q_{z}(x) \le \frac{\delta}{T^2} T = \frac \delta T.
\]
Hence, by Markov's inequality,
\begin{equation}\label{eventB}
\underset{R}{\mathbb{P}}\left[\sum_{z \in R\setminus\{x\}}q_{z}(x) \ge \frac cT\right] \le \frac Tc\cdot \frac\delta T  = \frac \delta c\;.
\end{equation}
Let $J$ denote the subset of $x\in I$ that satisfy (A) from~\eqref{AandB} and do not satisfy (B).
As events (A) and (B) are independent, the probability that $x\in I$ satisfies $x\in J$ is at most $\delta^2/(cT^2)$.  Hence, by Markov's inequality,
\begin{equation} \label{J}
\underset{R}{\mathbb{P}}\left[|J| \le \frac{10|I|\delta^2}{c T^2}\right] \ge 0.9.
\end{equation}
From~\eqref{H} and~\eqref{J}, we get that, with probability at least 0.8,
\[
|G| = |H| - |J| \ge \frac{|I|\delta}{2T^2} - \frac{10|I|\delta^2}{c T^2} 
\ge \frac{\eps'N}{T^2} \left( \frac\delta2 - \frac{10\delta^2}{c} \right) 
= \Omega\left(\frac{\eps N}{T^2} \right),
\]
if $\delta$ is a small enough positive constant.
\end{proof}

We now describe an encoding assuming $|G| = \Omega(\eps N/T^2)$. 
We assume that the encoding procedure fails if $|G|$ is smaller, hence, its success probability is $0.8$.
Before proceeding, note that if we want to encode a set $Q \subseteq [N]$ containing $k$ elements, it takes $\log\binom{N}{k}$ bits to specify it since for each of the $\binom{N}{k}$ subsets of $[N]$ of size $k$ we can assign a binary string to them, thereby specifying $Q$. Similarly, for any permutation $g$ on $N - u$ elements (where $0 \le u \le N-1$), we can specify it using $\log(N - u)!$ bits. Our encoding contains the following information:
\begin{itemize}
\item The advice string $\alpha$;
\item The cardinality of the set $G$ of good elements;
\item The set $f(R)$, encoded using $\log\binom{N}{|R|}$ bits;
\item The values of $f$ restricted to $f: [N]\setminus R \to [N]\setminus f(R)$, encoded using $\log(N - |R|)!$ bits;\footnote{That is, this part of the encoding is a permutation $g : [N - |R|]\to [N - |R|]$, with the meaning that if $g(i) = j$, then $f$ maps the $i$-th element of the set $[N]\setminus R$ to the $j$-th element of the set $[N] \setminus f(R)$. Note that knowledge of the sets $R$ and $f(R)$ is needed to decode this part of the encoding. This will not be a problem because the decoder knows $R$, which is part of the common random string, and is given $f(R)$.}
\item The set $f(G)$ of images of good elements of $R$, encoded using $\log\binom{|R|}{|G|}$ bits;
\item The values of $f$ restricted to $f: R\setminus G \to f(R\setminus G)$, encoded using $\log(|R|-|G|)!$ bits.\footnote{Similar remarks hold as made in the previous footnote. The decoder needs to know the sets $R \setminus G$ and $f(R\setminus G)$ to decode this part of the encoding. Although we have not explicitly specified the set $G$, the decoder will reconstruct $G$ from the encoding, as described in the decoding procedure.}
\end{itemize}

The decoding procedure is as follows.
It initializes an empty table to store the values of $f$, and it fills up the mapping from $[N] \setminus R$ to $[N]\setminus f(R)$.  It can be done, as the decoding procedure knows both $R$ (as the randomness of the encoder), $f(R)$ and the mapping (the latter two from the encoding).
Next, it inverts all $y\in f(G)$ in the way we are about to describe.
In particular, it determines the set $G$.
After that, it decodes the mapping from $R\setminus G$ to $f(R\setminus G)$.

Thus, it suffices to describe how to invert $y\in f(G)$.
Let $x=f^{-1}(y)$.  The decoder wants to find $x$.
Consider the oracle $h$ given by 
\begin{equation*}
h(z) =
\begin{cases}f(z), &\mbox{  if $z \in [N]\setminus R$;}\\ y, &\mbox{ if $z \in R$.}
\end{cases}
\end{equation*}
The decoder has enough information to simulate oracle $h$.
Also, $h(z) = f(z)$ for all $z\in ([N]\setminus R)\cup \{x\}$.
Now let $\mathcal{A}^{f}(y)$ denote $\mathcal{A}$ with advice $\alpha$, oracle access to $f$, and input $y$; and let $\mathcal{A}^{h}(y)$ denote $\mathcal{A}$ with advice $\alpha$, oracle access to $h$, and input $y$. 
Let $\ket{\phi_f}$ and $\ket{\phi_h}$ denote that final states of $\mathcal{A}^{f}(y)$ and $\mathcal{A}^{h}(y)$, respectively. 
Then by Lemma~\ref{swap-lemma} and the definition of a good element,
\begin{equation*}
\|\ket{\phi_f}-\ket{\phi_h}\| \le \sqrt{T\sum\nolimits_{z \in R\setminus \{x\}}q_{z}(x)}
\le \sqrt{T\cdot \frac{c}{T}} = \sqrt{c}.
\end{equation*}
As $x\in I$, measuring $\ket{\phi_f}$ gives $x$ with probability at least $2/3$.  
If $c$ is small enough, measuring $\ket{\phi_h}$ gives $x$ with probability strictly greater than $1/2$.
Thus, the decoding procedure can determine $x$ with certainty by deterministically simulating the quantum procedure ${\cal A}^h(y)$.

Finally, the length of our encoding is
\begin{equation*}
\begin{split}
&S + \log\left(\frac{N!}{(N-|R|)!|R|!}\cdot (N-|R|)!\cdot\frac{|R|!}{(|R|-|G|)!|G|!}\cdot (R-|G|)!\right) + O(\log N)\\
&= S + \log N! - \log|G|! + O(\log N).\qedhere
\end{split}
\end{equation*}
\end{proof}

From Lemmas~\ref{l1} and~\ref{l2}, we get that
\[
\log \left(\frac{\eps N!}4 \right) \le \log N! - \Omega\left(\frac{\eps N}{T^2}\right) + S + O(\log N),
\]
which easily implies that $T^2\cdot S = \tilde\Omega(\eps N)$.

\subsection{Quantum upper bound?}
Having established the lower bound in Theorem~\ref{result2}, the question then arises: can we provide a quantum algorithm such that for every choice of parameters $S, T$ such that $T^2\cdot S = N$, there is a data structure of size $\tilde{O}(S)$ and an algorithm that with the help of the data structure, given $f(x)$ for any $x \in [N]$, can invert $f$, which is a permutation, in time $\tilde{O}(T)$?\newline
\indent Unlike in the box problem where we could change the advice to be different than Yao's advice of parity, it seems unlikely that any classical advice other than iterates of $f$ would be useful for solving this problem. Moreover, the bottleneck seems to be the speed at which one can iterate $f$. However, by the results of Ozhigov \cite{ozhigov} and Farhi et al. \cite[pg. 5444]{farhi}, we have that function iteration cannot be sped up by a quantum computer (namely, in order to compute $f^{(s)}(x)$ we cannot do better than $O(s)$ applications of $f$).\newline
\indent It might be possible to obtain a speedup by using Grover's algorithm to find a starting point for the function iteration. For the sake of simplicity, suppose $f$ is a cyclic permutation. Given a $y \in [N]$ that we want to invert, one idea would be to have the advice be $N^{0.5-\eps}$ ``equally spaced'' points so that each interval between the given points is of size $N^{0.5 + \eps}$, for some fixed $\eps > 0$. \newline
\indent Then the problem of inverting $y$ essentially reduces to the problem of finding the closest advice point to $y$. Now, if we are able to find some $y'$ that lies halfway between the two advice points that $y$ lies in between of, then we can recurse on an interval that is half as long to find the preimage of $y$. To find this intermediate point $y'$, we can employ randomness. If we pick $N^{0.5-\eps}$ random points, then on average, one of these points will lie in the interval that $y$ is in. Now, if it was an efficiently computable property that this point lies halfway between the two advice points that $y$ lies in between of, then we could find the right point among the randomly chosen points in time $O\left(\sqrt{N^{0.5-\eps}}\right)$ using Grover's algorithm. \newline
\indent However, checking that this property holds seems to inevitably involve iterating $f$, making it unlikely that it is efficiently computable. But suppose even that this property is efficiently computable, then the running time grows linearly with the number of iterations of $f$ one needs to perform, but only in the square root of the number of points. As a result, if we pick $k < N$ points at random, then it seems that we would still need around $N/k$ iterations, so the running time would be approximately $O(\sqrt{k}\cdot (N/k)) = O(N/\sqrt{k})$, which will always be worse than $O\left(N^{0.5}\right)$.

\section{Conclusions and further work}
We originally proved the result of Theorem~\ref{result1} in order to shed light on how to eventually prove Theorem~\ref{result2}. We turned to the hybrid method, as it was enough to prove Theorem~\ref{result1}. Moreover, it turned out to also be an applicable lower bound technique to proving Theorem~\ref{result2}. Now, the advice we considered here was classical, as that was the same advice used by De et al. \cite{de} and seemed most intuitive to reason about. However, advice given to a quantum algorithm can also be a quantum state. The difficulty here is we cannot copy quantum advice (due to the ``no cloning'' theorem of Wootters and Zurek \cite{no-clone}). Therefore, we cannot adapt our current proof for permutation inversion given in \S 4.2 directly because in Lemma~\ref{l2} we have to get the elements of $G$ repeatedly with the same piece of advice.\newline
\indent Another direction would be to prove or disprove the optimality of the lower bounds obtained in either Theorem~\ref{result1} or Theorem~\ref{result2}. On the lower bounds side, perhaps our results can be improved by some variant of the standard lower bound techniques (since as discussed in \S 1, the techniques of Nishimura and Yamakami \cite{NY} cannot be directly applied to the problem of permutation inversion). \newline
\indent The best technique might be the adversary method (as this improves on the hybrid method). However, the adversary method is difficult to apply since, even in the simpler case of the box problem, the set $D_{\alpha}$ is arbitrary, so constructing a binary relation $R$ on the sets of inputs on which $f_{j,\alpha}$ differs, as required to apply the adversary argument by Theorem 6.1 of \cite{ambainis}, is not obvious. \newline
\indent Another line of attack would be to show that either the box problem reduces to solving parity, or that inverting a random permutation with advice reduces to solving function iteration, both of which are difficult for quantum computers. This would provide strong evidence for quantum computers not being able to solve either of these problems better than in the classical case. On the algorithms side, it would be interesting to see if one can improve the quantum upper bound for the box and permutation inversion problems with classical advice. If there is a quantum speedup, most likely it will be a novel algorithm, since simply using Grover's search algorithm \cite{grover} as a subroutine (or with randomness) does not seem to help. We suspect that with classical advice one cannot do better than the classical upper bounds (given the discussions in \S 3.4 and \S 4.3).
\newpage
\section*{Acknowledgements}
A. N. would like to thank Mark Zhandry for several helpful discussions. S. A. is supported by an Alan T. Waterman Award from the National Science Foundation. A. B. is supported by FP7 FET Proactive project QALGO. Part of this work was done while A. B. was at CSAIL, Massachusetts Institute of Technology, USA, supported by Scott Aaronson's Alan T. Waterman Award from the National Science Foundation. L. T. is supported by the National Science Foundation under Grant No. 1216642 and by the US-Israel Binational Science Foundation under Grant No. 2010451.

\bibliographystyle{amsplain}

\end{document}